%% file: main.tex
\pdfoutput=1
\documentclass[runningheads,a4paper]{llncs}

\usepackage{amsmath,amssymb,latexsym}
\usepackage{amssymb}
\usepackage{amsfonts}

\usepackage[ruled]{algorithm}
\usepackage{algpseudocode}
\usepackage{listings}

\usepackage{graphicx}
\usepackage{caption}

\usepackage{multirow}
\usepackage[normalem]{ulem}

\usepackage{tikz} 
\usetikzlibrary{arrows,automata,shapes,positioning,calc, matrix} 
\usepackage{circuitikz}

\usepackage{url}
\urldef{\mailne}\path|paolo.zuliani@ncl.ac.uk|
\urldef{\mailud}\path| anticoli.linda@spes.uniud.it, carla.piazza@uniud.it|
\newcommand{\keywords}[1]{\par\addvspace\baselineskip
\noindent\keywordname\enspace\ignorespaces#1}

\begin{document}

\mainmatter  

\title{Verifying Quantum Programs: \\
From Quipper to QPMC\thanks{This work has been partially supported by the GNCS group of INdAM.}}

\titlerunning{Verifying Quantum Programs: From Quipper to QPMC}

%
%
\author{Linda Anticoli\inst{1}  \and Carla Piazza\inst{1} \and Leonardo Taglialegne\inst{1} \and Paolo Zuliani\inst{2}}
\authorrunning{Linda Anticoli \and Carla Piazza  \and Leonardo Taglialegne \and Paolo Zuliani}

\institute{
Dept. of Mathematics, Computer Science and Physics,
University of Udine, Italy
\mailud
\and
School of Computing Science,
Newcastle University, 
United Kingdom
\mailne}

%
%

\maketitle

\input{macro}

\begin{abstract}

In this paper we present a translation from the quantum programming language Quipper to the
QPMC model checker, with the main aim of verifying Quipper programs.
Quipper is an embedded functional programming language
for quantum computation. It is above all a circuit description language,
for this reason it uses the vector state formalism and its main purpose is
to make circuit implementation easy providing high level operations for
circuit manipulation. Quipper provides both an high-level circuit building interface and a simulator.
QPMC is a model checker for quantum protocols based on the
density matrix formalism. QPMC extends the probabilistic model checker IscasMC allowing to formally
verify properties specified in the temporal logic QCTL on Quantum Markov Chains.
We implemented and tested our translation on several quantum algorithms, including Grover's
quantum search.

\keywords{Quantum Languages, Quantum Circuits, Model Checking}
\end{abstract}

\input{introduction}
\input{preliminaries}

\input{translation}

\input{implementation}
\input{experiments}

\input{conclusion}

\bibliographystyle{
plain}
\bibliography{bibliography}

\input{appendix}

\end{document}

%% file: macro.tex
\def\va{\mathbf{a}}
\def\vb{\mathbf{b}}
\def\vc{\mathbf{c}}
\def\vd{\mathbf{d}}
\def\vf{\mathbf{f}}
\def\vg{\mathbf{g}}
\def\vk{\mathbf{k}}
\def\vi{\mathbf{i}}
\def\vp{\mathbf{p}}
\def\vs{\mathbf{s}}
\def\vu{\mathbf{u}}
\def\vv{\mathbf{v}}
\def\vx{\mathbf{x}}
\def\vy{\mathbf{y}}
\def\vz{\mathbf{z}}

\def\natural{\mathbb{N}}
\def\real{\mathbb{R}}
\def\bool{\mathbb{B}}

\def\true{true}
\def\false{false}

\def\ParaSet{P}
\def\PredSet{\Sigma}
\newcommand{\ReachSet}[2]{\mathcal{R}^{#1}_{#2}}
\newcommand{\Behavior}[2]{\mathcal{W}^{#1}_{#2}}

\def\pred{\sigma}
\newcommand{\until}[1]{\mathcal{U}_{#1}}
\newcommand{\release}[1]{\mathcal{R}_{#1}}
\newcommand{\eventually}[1]{\Diamond_{#1}}
\newcommand{\always}[1]{\Box_{#1}}

\newcommand{\CharFun}[3]{\mathcal{X}(#1,#2,#3)}

\newcommand{\ParaLS}[1]{A_{#1} \vp \leq b_{#1}}

%% file: introduction.tex

\section{Introduction}\label{sec:introduction}
The specification of algorithms in human readable form and their translation into machine executable code is one of the main goals of high-level programming languages. Quantum algorithms and protocols are usually described by quantum circuits (i.e., circuits involving quantum states and quantum logic gates). Even if such circuits have a simple mathematical description they could be very difficult to realise in practice without a deep knowledge of the essential features of the physical phenomena under consideration.  The above reasons justify the need for tools that permit to abstract from a low-level description of quantum algorithms and protocols allowing also people that know very little of quantum physics to program a quantum device.

The introduction of high-level formalisms allows to define and automatically verify formal properties of algorithms abstracting away from low-level physical details. Formal techniques are an important tool for the validation and verification of programs in classical computer science. Experimental verification (testing) could be done, but there is no assurance that each possible error is avoided. With formal verification techniques such as \emph{model checking} we can test temporal properties of an algorithm evaluating all possible cases. In the context of quantum computation, that is based on the counter-intuitive laws of quantum physics, the possibility of testing quantum protocols is very important. In particular, protocols for quantum cryptography, that are deeply investigated at the moment hoping for future applications in the secure transmission of information, require certifications of correctness. 

Although both the quantum computation and verification fields are quite new, we found two interesting tools: the functional language Quipper \cite{Selinger2014} and the model checking system QPMC \cite{QPMC}, which we decided to use as a starting point for the development of a framework providing both a high-level programming style and formal verification tools. In particular, Quipper is a quantum programming language based on Haskell that allows to build quantum circuits by describing them in a simple programming style and provides the possibility to simulate the circuit. QPMC is a model checker for quantum protocols that uses an extension of PCTL, a probabilistic temporal logic, to verify properties of quantum protocols.  Quipper has been used to program a set of non-trivial quantum algorithms, it is supported by a community and provides a high-level programming environment based on Haskell. Unfortunately Quipper lacks of a built-in formal verification tool. On the other hand, QPMC supports formal verification but it is based on a low-level specification language. Hence, we decided to build a bridge between them, translating Quipper code into the QPMC formalism, thus providing an ad-hoc verification framework to Quipper programmers. 

This is just a first step in the direction of providing a complete programming framework for quantum computing. As the authors point out in \cite{QPMC}, QPMC is intended for verification
of \emph{classical} properties for which only the measurement outcomes
as well as the probabilities of obtaining them are relevant. Quantum
effects caused by superposition, entanglement, etc., are merely employed
to increase the efficiency or security of the protocol. Hence, more sophisticated logical formalisms allowing to specify/verify quantum effects and more in general reversible computation properties should be introduced.

The paper is organized as follows. In Section \ref{sec:preliminaries} we recall some basic quantum notations and briefly introduce Quipper and QPMC. In Section \ref{sec:translation} we define an abstract algorithm for translating Quipper circuits int Quantum Markov Chains, i.e., QPMC models. In Section \ref{sec:implementation} we describe our implementation of the translation algorithm. In Seciton \ref{sec:experiments} we discuss some experimental results on the verification of Grover's algorithm and some scalability tests. Section \ref{sec:conclusion} ends the paper.

%% file: preliminaries.tex

\section{Preliminaries}\label{sec:preliminaries}

\subsection{Mathematical Quantum Models}\label{sec:quantum}

Quantum systems are represented through complex Hilbert spaces. A complex Hilbert space $\mathcal H$ is a complete vector space equipped with an inner product inducing a complete metric space. In particular, we will consider quantum systems described by finite dimensional Hilbert spaces of the form $\mathbb C^{2^k}$.
The elements of $\mathcal{H}$ (vectors) are denoted by either $\psi$ or $|\psi\rangle$ (i.e., ket notation). The notation $\langle \psi |$  (i.e., bra notation) denotes the transposed conjugate of $|\psi \rangle$.
The scalar product of two vectors $\varphi$ and $\psi$ in $\mathcal{H}$ is denoted by $\langle\varphi|\psi\rangle$, whereas $|\varphi\rangle \langle \psi |$ denotes the linear operator 
defined by $|\varphi\rangle$ and $\langle \psi|$.  We use $I$ to denote the identity matrix and $tr(\cdot)$ for the matrix trace. 

There are two possible formalisms based on Hilbert spaces for quantum systems: the \emph{state vector} formalism and the \emph{density matrix} one. We briefly introduce both of them, since Quipper is based on state vectors, while QPMC exploits density matrices.

\subsubsection{State Vector Formalism}


The \emph{state} of a quantum system is described by a \emph{normalized vector}  $|\psi\rangle \in \mathcal{H}$, i.e., $\lVert|\psi\rangle\rVert = \sqrt{\langle\psi|\psi\rangle}= 1$. 
The normalization condition is  related to the probabilistic interpretation of quantum mechanics. 

The temporal \emph{evolution} of a quantum system is described by a \emph{unitary operator} (see, e.g., \cite{Nielsen-Chuang}). 
In particular, a linear operator $U$ is unitary if  its conjugate transpose $U^{\dag}$  coincides with its inverse $U^{-1}$. 
Unitary operators preserve inner products and, as a consequence, norms of vectors. 
In absence of any measurement process, the state $|\psi_0\rangle$ at time $t_0$ evolves at time $t_1$ through the unitary operator $U$ to the state
$$|\psi_1\rangle = U \ |\psi_0\rangle$$

An \emph{observable} is a property of a physical system that can be measured, i.e., a physical quantity such as energy, position, spin. Observables are described by \emph{Hermitian 
operators} (see, e.g., \cite{Preskill}). 
A linear operator $A$ is Hermitian if $A = A^{\dag}$.  
Assuming non degeneracy, an 
Hermitian operator $A$ can be decomposed as 
$$A=\sum_{i=1}^n a_i | \varphi_i\rangle \langle \varphi_i|$$ where the $a_i$'s ($|\varphi_i\rangle$'s) are the eigenvalues (eigenvectors, respectively) of $A$.
The eigenvalues of a Hermitian operator are real.
 Given a system in a state $|\psi\rangle$,
the outcome of a measurement of the observable $A$ is one of its eigenvalues $a_i$ and 
the state vector of the system after the measurement is
\begin{equation*}\frac{(|\varphi_i\rangle\langle \varphi_i|) |\psi\rangle}{||(|\varphi_i\rangle\langle \varphi_i|) |\psi\rangle ||}\end{equation*}
with probability
\begin{equation*}p(a_i)  = || (|\varphi_i\rangle\langle \varphi_i|) |\psi\rangle ||^2 = \langle\psi|(|\varphi_i\rangle\langle \varphi_i|) |\psi\rangle\end{equation*}


\subsubsection{Density Matrix Formalism}\label{sec:density}
Here density matrices take the role of state vectors. However, the states described by state vectors on Hilbert spaces are idealized descriptions that cannot characterize statistical (incoherent) mixtures which often occur in Nature. Density matrices allow to represent also such mixed states.


A matrix $\rho$ is positive if for each vector $|\psi\rangle$ it holds that $\langle\psi|\rho|\psi\rangle \geq 0$. 
The \emph{state} of a quantum system is described by an Hermitian, positive matrix $\rho$ with $tr(\rho)=1$.
Such matrices are called \emph{density} matrices.

Given a normalized vector $|\psi\rangle$ representing the state of a system through the state vector formalism, the correspondent density matrix is $|\psi\rangle \langle \psi|$.


\emph{Evolutions} and \emph{measurements} of quantum systems are now described by \emph{superoperators} \cite{Nielsen-Chuang}. 
A superoperator is a (linear) function $\mathcal{E} : \rho_0 \to \rho_1$ which maps a
density matrix $\rho_0$ at time $t_0$ to a density matrix $\rho_1$ at time $t_1>t_0$
that satisfies the following properties:
$\mathcal{E}$ preserves hermiticity; 
$\mathcal{E}$ is trace preserving;
$\mathcal{E}$ is completely positive.

Given a unitary operator $U$ the corresponding superoperator $SO(U)$ can be defined as follows:
$$SO(U)(\rho_0) = U\rho_0 U^{\dag}$$

A quantum \emph{measurement} is described by a collection $\{M_i\}$ of measurement operators satisfying the following condition
\begin{equation*}\sum_i{M_i^\dag M_i} = I\end{equation*} 
The index \emph{i} refers to the possible measurements outcomes. If $\rho$ is the state before the measurement, then the result is \emph{i} and the state after the measurement is 
\begin{equation*}\frac{M_i \rho M_i^\dag}{tr(M_i \rho M_i^\dag)}\end{equation*}
with probability
\begin{equation*}p(i) = tr(M_i \rho M_i^\dag)\end{equation*}

Given an observable $A=\sum_{i=1}^n a_i | \varphi_i\rangle \langle \varphi_i|$ in the state vector formalism, its correspondent in the density matrix one is $\{|\varphi_i\rangle \langle \varphi_i|\}$.

\subsection{Quipper}\label{sec:quipper}

Quipper is an embedded functional programming language for quantum
computation \cite{Selinger2013} based on the Knill's QRAM model \cite{knill-qram} of quantum computation. 
This model uses both a classical and a quantum device to perform a quantum computation. The classical device performs classical computations (control flow, test, loops) and the
quantum computer is a specialised device that is able to perform only
two kinds of instruction: unitary operations and measurements.

Quipper has a collection of data types,
combinators, and a library of functions within Haskell, together with
an idiom, i.e., a preferred style of writing embedded programs \cite{Selinger2013}.
It provides an extended circuit model of quantum computation which is concerned with qubits 
 and unitary gates 
 and allows also
classical wires (whose state is a classical bit) and gates within
a circuit.

Quipper is above all a circuit description language, for this reason
it uses the state vector formalism and its main purpose is
to make circuit implementation easier providing high level operations for circuit manipulation.
A Quipper program is a function that inputs some quantum and classical data, performs
state changes on it, and then outputs the changed quantum/classical data. This
is encapsulated in a Haskell monad called \texttt{Circ}, which from an abstract point of view returns a quantum circuit.
The philosophy of the Quipper paradigm
is that qubits are held in variables and gates are applied to them
one at a time. 
A set of predefined gates (e.g., \texttt{hadarmard}, \texttt{cnot}, \dots), together with the possibility of specifying \emph{ancilla} qubits and \emph{controls}, are provided.

In this paper we focus on the \texttt{Circ} monad of Quipper, where a sequence of unitary and measurement gates can be applied to qubits and bits.
Quipper allows to generate a graphical representation and to simulate through three different simulators a circuit written in the monad. 
In Figure \ref{fig:example_1}
we show the graphical representation of a simple quantum circuit in which the Hadamard gate is applied to one qubit. Such circuit is defined in Quipper through the following code.

\begin{lstlisting}[language=Haskell, breaklines=true, firstline=15, basicstyle=\scriptsize\ttfamily\linespread{0.5}, lastline=18]
oneq :: Qubit -> Circ Qubit
oneq q = do
  hadamard_at q
  return q
\end{lstlisting}
\vspace{-8ex}
\begin{figure}[H]
\begin{center}\includegraphics[scale=0.30]{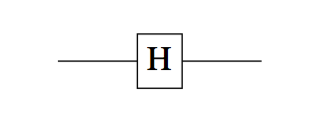}\end{center}
\protect\caption{One qubit circuit }
\label{fig:example_1}
\end{figure}

\vspace{-6ex}

\subsection{QPMC: Quantum Program/Protocol Model Checker}

QPMC is a model checker for quantum programs and protocols based on the density matrix formalism
available in both web-based and off-line version at http://iscasmc.ios.ac.cn/too/qmc.
It takes in input programs written in
an extension of the guarded command language
PRISM \cite{PRISM} that permits, in addition to the constants definable
in PRISM, the specification of types \texttt{vector}, \texttt{matrix}, and \texttt{superoperator}.
QPMC supports the bra-ket notation and inner, outer and tensor product can be written using it.

The semantics of a QPMC program is given in terms of 
\emph{superoperator weighted Markov chain}, which is a Markov chain in which the
state space is taken classical, while all quantum effects are encoded in the superoperators labelling the transitions 
(see, e.g., \cite{QPMC,MCQMC2013}). 
Differently from what we defined in Section \ref{sec:density}, QPMC superoperators are not necessarily trace-preserving, they are just completely positive linear operators.
A trace-non-increasing superoperator describes processes in which extra-information is obtained by \emph{measurement}. 
Let $\mathcal{S(H)}$ be the set of superoperators over a Hilbert
space $\mathcal{H}$ and $\mathcal{S^I}(\mathcal{H})$ be the subset of
trace-nonincreasing superoperators.
Given  a density matrix $\rho$ representing the state of a system, 
$\mathcal{E} \in \mathcal{S^I(H)}$ implies that 
$tr(\mathcal{E}(\rho)) \in [0,1]$. Hence, it is natural to regard the set $\mathcal{S^I}(\mathcal{H})$ as the
quantum correspondent of the domain of traditional probabilities \cite{QPMC}.

Let $\mathcal{E}, \ \mathcal{F} \in \mathcal{S(H)}$ we say that $\mathcal{E} \mathbf{\lesssim} \mathcal{F}$ if for any quantum state $\rho$ it holds that $\ tr(\mathcal{E}(\rho)) \leq tr(\mathcal{F}(\rho))$. 
A QMC is a discrete time Markov chain, where classical probabilities are replaced with quantum probabilities. 
\begin{definition}[Quantum Markov Chain \cite{MCQMC2013,QPMC}] A superoperator
weighted Markov chain, also referred to as quantum Markov chain (herein
QMC) over a Hilbert space $\mathcal{H}$ is a tuple $\mathit{(S, Q, AP,
L)}$, where:
\begin{itemize}
\item $S$ is a countable (finite) set of classical states; 
\item $Q: S \times S \rightarrow \mathcal{S^I(H)}$
is called the transition matrix where for each $s \in S$, the superoperator $\sum_{t\in S} {Q(s,t)}$ is trace-preserving
\item $AP $ is a finite set of atomic propositions
\item $L : S \to 2^{AP}$ is a labelling function
\end{itemize}
\end{definition}

The aim of QPMC is to provide a formal framework where to define and analyse properties of quantum protocols. 
The properties to be verified over QMC are expressed using the quantum computation tree logic
(QCTL), a temporal logic for reasoning about evolution of quantum
systems introduced in \cite{MCQMC2013} that is a natural extension
of PCTL
.
\begin{definition}[Quantum Computation Tree Logic \cite{MCQMC2013,QPMC}]
A QCTL formula is a formula over the following grammar:
\begin{center}$\Phi ::= a \ | \  \lnot\Phi \ | \ \Phi \wedge \Phi \ | \ \mathbb{Q_{\sim \epsilon}}[\Phi] $\end{center}\begin{center}$\phi ::= X\Phi \ | \ \Phi U^{\leq \mathit{k}} \Phi \ | \ \Phi U \Phi $\end{center}
where $a \in AP$, $\sim \ \in \{ \lesssim, \gtrsim, \eqsim \}$, 
$\mathcal{E} \in \mathcal{S^I}(\mathcal{H})$, $k \in \mathbb{N}$. 
$\Phi$ is a \emph{state }formula, while $\phi$ is a \emph{path } formula. 
\end{definition}
The quantum operator formula $\mathbb{Q_{\sim \epsilon}}[\phi]$ is
a more general case of the PCTL probabilistic operator $\mathbb{P_{\sim p}}[\phi]$
and it expresses a constraint on the probability that the paths from
a certain state satisfy the formula $\phi$.
Besides the logical operators presented in QCTL, QPMC supports an
extended operator $Q=? [\phi]$ to calculate (the matrix representation of) the superoperator
satisfying $\phi$. Moreover,
QPMC provides a function $qeval((Q=?)[\varphi], \rho)$ to compute the density operator obtained from applying
the resultant superoperator on a given density operator $\rho$,
and $qprob((Q=?)[\phi], \rho) = tr(qeval((Q=?)[\phi], \rho)))$ to calculate the probability of satisfying $\phi$,
starting from the quantum state $\rho$ \cite{QPMC}.


%% file: translation.tex

\section{From Circuits to Quantum Markov Chains}\label{sec:translation}

In order to be able to define a mapping from Quipper to QPMC programs in this section we work at the semantic level. This means that we consider a quantum circuit generated by Quipper and we define a correspondent QMC having an \emph{equivalent} behavior.

\subsection{Circuits}

We first need a formal definition of quantum circuits generated from Quipper. Even though Quipper supports also classical wires, here we focus on circuits over quantum ones.  
As in Quipper, we consider only measurements of one qubit at a time with respect to the standard computational basis.
We assume the reader to be familiar with the classical notions of graphs and boolean circuits. Given a node $v$ of a directed graph we use the notation $In(v)$ ($Out(v)$) to denote the 
number of edges incoming (outcoming, respectively) in $v$. A quantum circuit is an extension of a boolean circuit in which operation gates are labeled with unitary operators.
When a unitary operator is applied to $k$ qubits it is necessary to know in which order the qubits are used for this reason each edge of a quantum circuit has two integer labels. 
\begin{definition}[Quantum Circuit]
A Quantum Circuit is a directed acyclic graph (herein DAG) $C=(V,E)$ whose nodes, also called \emph{gates}, are of types \emph{Qubit (Q)}, \emph{Unitary (U)}, \emph{Measurement (M)} and \emph{Termination (T)} and satisfy the following conditions:
\begin{enumerate}
\item Q gates: each node $v$ of type \emph{Qubit} is an input node, i.e. $In(v) = 0$ and $Out(v) = 1$;
\item U gates: each node $v$ of type \emph{Unitary} is labelled with an integer $dim(v)$ and a square unitary matrix $U(v)$ of complex numbers of dimension $2^{dim(v)}$. Moreover, it holds that $In(v) = Out(v) = dim(v)$;
\item M gates: each node $v$ of type \emph{Measurement} is an output node, i.e. $In(v)=1$ and $Out(v)=0$;
\item T gates: each node $v$ of type \emph{Termination} is an output node, i.e. $In(v)=1$ and $Out(v)=0$.
\item Edges: each edge $e\in E$ is labelled with two integers $\mathcal{S}(e)$ and $\mathcal{T}(e)$ such that: 
\begin{itemize}
\item for each node $u$ the set of labels $\mathcal{T}(\cdot)$ of the edges ingoing in $u$ is $\{1,\dots,In(u)\}$;
\item for each node $u$ the set of labels $\mathcal{S}(\cdot)$ of the edges outgoing from $u$ is $\{1,\dots, Out(u)\}$.
\end{itemize}
\end{enumerate}
A Quantum Circuit with $k$ nodes of type \emph{Qubit} is said to have \emph{size} $k$.
\end{definition}

\begin{example}\label{ex1}

Let us consider the following Quipper function implementing Deutsch's algorithm.

\begin{lstlisting}[language=Haskell, breaklines=true, firstline=3, basicstyle=\scriptsize\ttfamily\linespread{0.5}, lastline=9]
deutsch :: (Qubit, Qubit) -> Circ Bit
deutsch (q1, q2) = do
    hadamard q1
    hadamard q2
    qnot_at q2 `controlled` q1
    hadamard q1
    measure q1
\end{lstlisting}

Quipper graphically represents the circuit as shown in Figure \ref{fig:quipper-deutsch}.

\begin{figure}[H]
\begin{center}\includegraphics[scale=0.35]{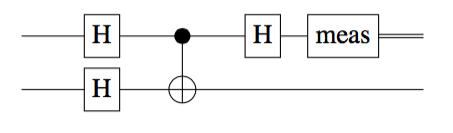}\end{center}
\protect\caption{Deutsch circuit in Quipper}
\label{fig:quipper-deutsch}
\end{figure}

Our definition enriches the above representation with labels denoting the order in which the qubits are used, as depicted in Figure \ref{fig:labels-deutsch}.

\vspace{-5ex}
\begin{figure}[H]
\begin{center}\includegraphics[scale=0.30]{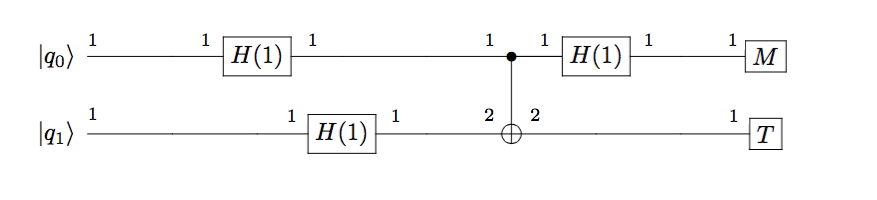}\end{center}
\protect\caption{Deutsch circuit with labels}
\label{fig:labels-deutsch}
\end{figure}
\end{example}

\vspace{-5ex}
\begin{definition}[Circuit Normal Form]
A Quantum Circuit of size $k$ is said to be in \emph{Normal Form} if each \emph{Unitary} node $v$ in the circuit has $dim(v)=k$.
\end{definition}

\vspace{-5ex}
\begin{figure}[H]
\begin{center}\includegraphics[scale=0.28]{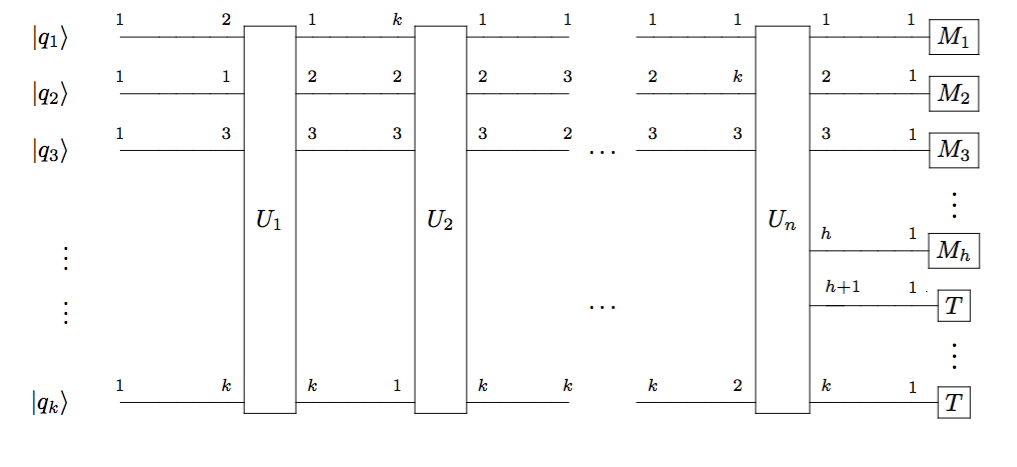}\end{center}
\protect\caption{Example of a circuit in Normal Form}
\label{fig:normal-form}
\end{figure}

\begin{definition}[Strong Normal Form]
A Quantum Circuit $C$ of size $k$ is said to be in \emph{Strong Normal Form} (herein SNF) if $C$ is in Normal Form, for each edge $e \in E$ between two Unitary nodes $\mathcal{S}(e) = \mathcal{T}(e)$ holds and the first $h\leq k$ edges outgoing the last Unitary node enter into Measurement nodes.
\end{definition}

\vspace{-5ex}
 \begin{figure}[H]
\begin{center}\includegraphics[scale=0.28]{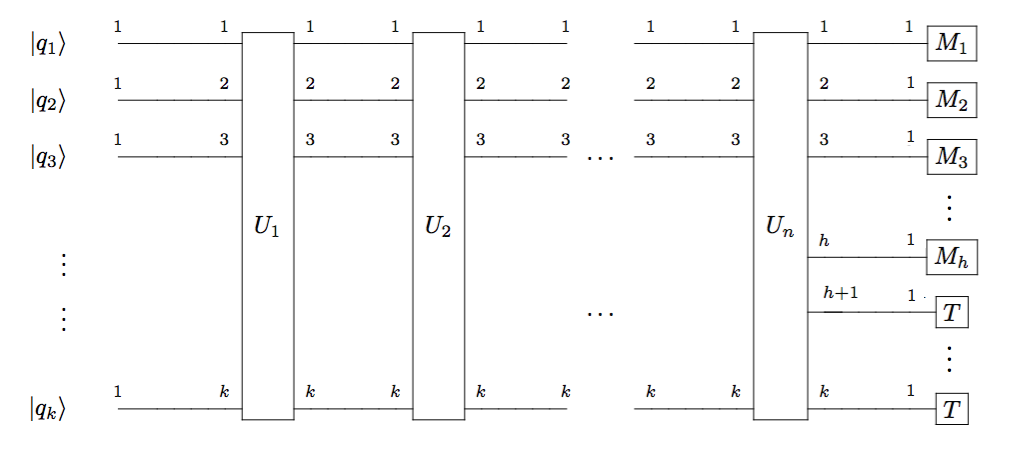}\end{center}
\protect\caption{Example of a circuit in Strong Normal Form}
\label{fig:strong-normal-form}
\end{figure}
A circuit $C$ in SNF is completely specified by the tuple $\langle k,[U_i,\dots,U_n],h \rangle$ where $k$ is the size of $C$, $U_1,\dots,U_n$ are the Unitary operators in the order they occur in $C$, and $h$ is the number of Measurement nodes.

%

In Figure \ref{fig:normal-form} we can see that in a circuit in Normal Form the order of the labels on the edges is not preserved. This is due to the fact that many gates require to be applied to a permutation of the input qubits. On the contrary, a circuit in SNF requires a precise ordering of the input and output edges. We will see that in order to match this requirement, SWAP operators have to be added. 

We now need a notion of equivalence between quantum circuits. This will allow us to move from a generic quantum circuit to a SNF circuit.
Intuitively, two quantum circuits are equivalent if, for any $k$-tuple of initial values of the qubits, the values of the qubits before measurements/terminations are the same. Moreover,
to be equivalent two circuits need to give the same outputs with the same probabilities.
Formally, let $C$ be a Quantum Circuit of size $k$ we denote by $Sem(C)$ the pair of functions $ (F(C),M(C))$ where: 
\begin{itemize}
\item $F(C):\mathcal{H}^k \longrightarrow \mathcal{H}^k$ is the function which maps $k$ qubits to the value they have just before the Measurement and Termination nodes;
\item $M(C):\mathcal{H}^k\times \{0,1\}^h \longrightarrow [0,1]$ is the function such 
that $M(C)(|\psi\rangle, (b_1,\dots,b_h))$ is the probability of getting output $(b_1,\dots,b_h)\in \{0,1\}^h$ on input $|\psi\rangle$.
\end{itemize}
Notice that if $C=\langle k,[U_i,\dots,U_n],h \rangle$ is a SNF circuit, then $F(C)(|\tau\rangle)=U_n\dots\ U_1 |\tau\rangle$. 
\begin{definition}[Quantum Circuit Equivalence]\\
Given two Quantum Circuits $C_1$ and $C_2$ of size $k$, $C_1$ and $C_2$ are equivalent, denoted by   $C_1 \approx C_2$ if and only if $ Sem(C_1) = Sem(C_2)$.

\end{definition}
%
%
\begin{lemma}
Every Quantum Circuit is equivalent to a circuit in Normal Form.
\label{NF}
\end{lemma}
\begin{proof}
Let $C$ be a Quantum Circuit of size $k$. $C$ is a DAG so it admits a \emph{topological ordering} of its nodes. Qubit nodes do not have any incoming edge so we choose an ordering in which the first nodes are all the ones of type Qubit. Measurement and Termination nodes do not have any outgoing edge so we choose them as final nodes in the ordering. The nodes in between initial and final nodes are only the one of type Unitary. We will proceed by induction on the number $n$ of Unitary nodes $\{U_1,\dots,U_n\}$.

\textbf{Base case: } For $n=1$ our circuit has only one Unitary gate. If $dim(U_1) = k$ then the circuit is in Normal Form.\\If $dim(U_1) = h < k$ we replace $U_1$ with the node $\widetilde{U}_1 = U_1 \otimes I$ where $dim(I) = {k-h}$ and then we append the remaining $k - h$ edges. 

\textbf{Induction step: } If $n>1$ by induction we know that we can normalise the first $n-1$ Unitary gates and we proceed as in the base case on the last one.
\qed
\end{proof}

\begin{lemma}\label{SNF}
Every Quantum Circuit in Normal Form is equivalent to a circuit in Strong Normal Form.
\end{lemma}
\begin{proof}
In order to prove our thesis we need a notion of generalised $\textit{SWAP}$ gate. The $\textit{SWAP}$ gate takes in input two qubits and swaps them, i.e., $\textit{SWAP} |x,y\rangle= |y,x\rangle$
A \emph{generalised} $\textit{SWAP}$ gate is an operator acting on $k$ qubits that returns in output a permutation $p_i$ of them. It is possible to build such operators by combining sequentially two-dimensional $\mathit{SWAP}$ gates.
Using the definition of generalised $\textit{SWAP}$ gates the proof is strightforward. Given a quantum circuit in Normal Form of size $k$, we obtain a circuit in SNC by opportunely swapping the Qubit indexes after the application of a unitary gate.
\qed
\end{proof}

\subsection{From Strong Normal Form Circuits to QMC's}

We are ready to define the QMC associated to a circuit in SNF. Intuitively, the states of the QMC correspond to the edges of the circuit, while the edges of the QMC connect subsequent states. Moreover, states without outgoing edges are added in the QMC to represent all the possible outputs of the circuit.
\begin{definition}[QMC associated to a Circuit]
Let $C$ be a Quantum Circuit in SNF of size $k$ with $n$ Unitary nodes $\{U_1, \dots ,U_n\}$ and $h$ Measurement nodes, the QMC $Q_C$ \emph{associated} to $C$ is defined as follows:
\begin{itemize}
\item the k-tuple of edges of $C$ entering the Unitary node $U_i$ is associated to the \emph{state} $s_i$ in $Q_C$;
\item the k-tuple of edges outgoing from the last Unitary node $U_n$ is associated to the state $s_{n+1}$;
\item in $Q_C$ there are $2^h$ states $t_0,t_1,\dots,t_{2^h-1}$;
\item for each $i \in \{1, \dots, n\}$ there is an edge from $s_i$ to $s_{i+1}$ is labelled with the superoperator $SO(U_i)$ associated to the Unitary gate $U_i$;
\item for each $i\in \{0, \dots , 2^h-1 \}$ there is an edge from $s_{n+1}$ to $t_i$ labelled with the superoperator $\widetilde{M}_i = M_i^h \otimes I^{k-h}$, where $I^{k-h}$ is the identity matrix of size $2^{k-h}$ and $M_i^h$ is a matrix of size $2^h$ having $1$ in the $i+1$-th position and all $0$'s in the remaining. 
\end{itemize}

\label{def:QC}
\end{definition}
In Figure \ref{fig:QMC} we can see the QMC associated to the example circuit shown in Figure \ref{fig:strong-normal-form}.
\begin{figure}[H]
\begin{center}\includegraphics[scale=0.275]{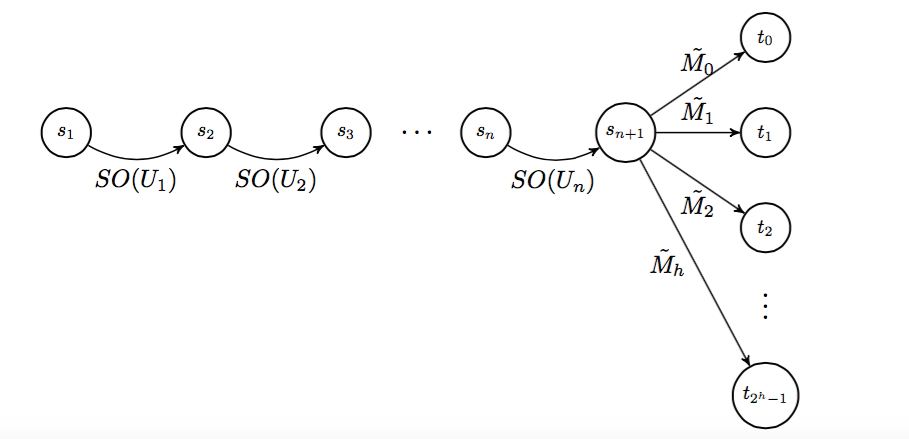}\end{center}
\protect\caption{QMC associated to the circuit  of Figure \ref{fig:strong-normal-form}}
\label{fig:QMC}
\end{figure}
\begin{lemma}\label{QMCQC}
Given a quantum circuit $C$ in SNF we can always build the QMC $Q_C$ associated to $C$ and it holds that:
\begin{enumerate}
\item $\forall |\tau\rangle \in \mathcal{H}, \ \forall i \in \{1,\dots, n\}$, \begin{gather*}U_i |\tau\rangle = |\psi\rangle \qquad \mbox{iff} \qquad SO(U_i)|\tau\rangle \langle \tau| SO(U_i)^{\dag} = |\psi\rangle \langle \psi|\end{gather*}
\item $\forall |\tau\rangle \in \mathcal{H}$ if $F(C)(|\tau\rangle)=|\psi\rangle$ and $M(C)(|\tau\rangle,\{b_1,\dots,b_h\})=p$, with $m=bin(b_1\dots b_h)$ (i.e., the natural with binary expansion $b_1\dots b_h$) then:
$$p = tr(\widetilde{M}_m |\psi \rangle \langle \psi| \widetilde{M}_m^{\dag})$$
and
$$\widetilde{M}_m |\psi\rangle \langle\psi| \widetilde{M}_m^{\dag} = |b_1,\dots,b_h\dots\rangle \langle b_1,\dots, b_h\dots|$$
\end{enumerate}
\end{lemma}
\begin{proof}
It immediately follows from our definitions.
\qed
\end{proof}
The above lemma states an equivalence between the semantics of a circuit $C$ in SNF and its associated QMC.
Thus, any Temporal Logic coherently defined on both formalisms can be equivalently model checked either on the circuit or on its associated QMC.
%
%
%

%% file: implementation.tex

\subsection{Translation Algorithm}
The results described in the previous sections allow us to define an algorithm that maps a quantum circuit into an \emph{equivalent} QMC.
In particular, Algorithm \textsf{Translate} performs the following steps:
\begin{itemize}
\item it transforms a quantum circuit into a normal form circuit (see Lemma \ref{NF});
\item it transforms a normal form circuit into a SNF circuit  (see Lemma \ref{SNF});
\item it transforms a SNF circuit into its corresponding QMC (see Definition \ref{def:QC}).
\end{itemize}
Hence, given a quantum circuit $C$ the output of \textsf{Translate}($C$) is a QMC equivalent to $C$ in the sense of Lemma \ref{QMCQC}.

The computational complexity of \textsf{Translate}($C$) depends on the number $n$ of Unitary nodes occurring in $C$ and on its size $k$.
For each Unitary we need to perform a number of binary swaps which depends on $k$. Without any efficient strategy in the worst case we could
perform $\Theta(k^2)$ binary swaps. Hence, \textsf{Translate}($C$) generates a QMC having $O(n*k^2)$ internal nodes. Each of this step requires
the computation of a matrix of size $2^k$.
However, we can lower the complexity of the algorithm by directly implementing generalized swaps without relying on binary ones. Such optimization
would generate a QMC having at most $O(n)$ internal nodes, requiring the computation of $O(n)$ swap matrices.

\section{Implementation}\label{sec:implementation}

In Section~\ref{sec:translation} we presented an abstract algorithm that translates a quantum circuit into a QMC.
We now describe an implementation of the \textsf{Translation} Algorithm in which the input quantum circuit is a Quipper function in the \texttt{Circ} monad and the
output QMC is a QPMC model.
%
Our implementation exploits the \texttt{Transformer} module of Quipper --a library providing functions for defining general purpose transformations on low-level circuits-- and works at data structure level. Using the \texttt{Transformer} module we can use Quipper's code, avoiding to implement the instructions again in an intermediate language.


The actual translation can be summarised in three steps. At first the gates in the quantum circuit must be grouped together with their associated qubits, taking care that the execution order is preserved. In this way we have an abstract representation of both the states and the transitions of the QMC. Then, as a second step, we calculate the matrix representation of the quantum gates. We also implemented a set of functions useful to perform operations on matrices (e.g., the tensor product). 
It is important to note that, since we need a circuit in SNF, our code provides a set of functions that generate the required swaps using compositions of binary swaps. Then our algorithm takes the resulting matrix and associate it to the gate input qubits, while the identity matrix is associated the remaining ones. Finally, the qubits are moved back in their original positions. All the matrices are computed in MATLAB notation. The last step is the conversion of the list of transitions into QPMC code. All these functions have been written in order to be kept as polymorphic as possible.
\begin{example}
Let us consider again the Quipper function for Deutsch Algorithm presented in Example \ref{ex1}.
As we showed, it can be compiled in Quipper generating the circuit represented in Figure \ref{fig:quipper-deutsch}. 



Our implementation converts Deutsch Quipper code into the QPMC model below.
\begin{lstlisting}[language=Haskell, breaklines=true, firstline=1, basicstyle=\scriptsize\ttfamily\linespread{0.5}, lastline=15]
qmc
const matrix A1 = [(1 / 2),(1 / 2),(1 / 2),(1 / 2);(1 / 2),(-1 / 2),(1 / 2),(-1 / 2);(1 / 2),(1 / 2),(-1 / 2),(-1 / 2);(1 / 2),(-1 / 2),(-1 / 2),(1 / 2)];
const matrix A2 = [1,0,0,0;0,1,0,0;0,0,0,1;0,0,1,0];
const matrix A3 = [(sqrt(2) / 2),0,(sqrt(2) / 2),0;0,(sqrt(2) / 2),0,(sqrt(2) / 2);(sqrt(2) / 2),0,((-1 * sqrt(2)) / 2),0;0,(sqrt(2) / 2),0,((-1 * sqrt(2)) / 2)];
const matrix A4 = [1,0,0,0;0,1,0,0;0,0,0,0;0,0,0,0];
const matrix A5 = [0,0,0,0;0,0,0,0;0,0,1,0;0,0,0,1];
module test
  s: [0..5] init 0;
  [] (s = 0) -> <<A1>> : (s' = 1);
  [] (s = 1) -> <<A2>> : (s' = 2);
  [] (s = 2) -> <<A3>> : (s' = 3);
  [] (s = 3) -> <<A4>> : (s' = 4) + <<A5>> : (s' = 5);
  [] (s = 4) -> (s' = 4);
  [] (s = 5) -> (s' = 5);
endmodule
\end{lstlisting}
\end{example}

Notice that, differently from what we wrote in our definition of QMC associated to a circuit, in the implementation we do not distinguish states $s_i$'s from states $t_i$'s in the generated QPMC model.

Our implementation is available at \url{https://github.com/miniBill/entangle}.

%% file: experiments.tex

\section{Experimental Results}\label{sec:experiments}
We have tested our translation tool with our Quipper implementation of \emph{Grover's search algorithm} \cite{Nielsen-Chuang}. The aim of Grover's algorithm is that of searching for the index $x$ of an element in a $N$-dimensional space with no structure. We assume $N = 2^n$, so that the indexes are represented by $n$-bit strings. The algorithm solves the problem by considering a function $f: \{0,1\}^n \to \{0,1\}$ such that $f(x) = 1$ if and only if the string $x$ is a solution. Classically, this problem can be solved in $O(N)$ steps while using a quantum oracle it can be probabilistically solved in $O(\sqrt{N})$ steps. 
Grover exploits quantum parallelism to give to the quantum oracle  all the possible input strings at the same time. 
Then the oracle marks the strings corresponding to possible solutions. At this point it performs some steps of amplitude amplification in order to maximize the probability of getting the desired result after the measurement. The result is the index of the searched element. The algorithm is probabilistic, because of the amplitude amplification step. 
Anyway, for $N = 4$, after one iteration it behaves in a deterministic way, giving the right result with probability equal to $1$.

For the experiment we decided to use a search space of size $N = 4$. The oracle returns the string $x = 3$, so the state after the measurement will collapse to $|11\rangle |1\rangle$. 
The algorithm needs an ancilla qubit that can be easily discarded at the end of the computation. 
The Quipper circuit of the algorithm can be seen in Figure \ref{fig:grover-algorithm}.
\begin{figure}[H]
\begin{center}\includegraphics[scale=0.35]{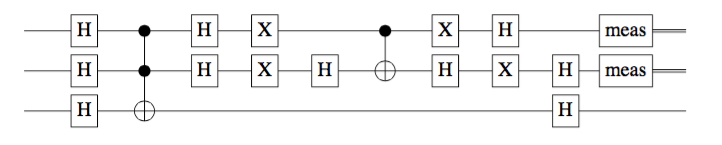}\end{center}
\protect\caption{Grover's algorithm in Quipper}
\label{fig:grover-algorithm}
\end{figure}

\subsubsection{Quipper Implementation}
At first we implemented the circuit using Quipper as shown below.
\begin{lstlisting}[language=Haskell, breaklines=true, firstline=1, basicstyle=\scriptsize\ttfamily\linespread{0.5}, lastline=19]
grover :: (Qubit, Qubit, Qubit) -> Circ (Bit, Bit)
grover (q1,q2,q3) = do
    hadamard_at q1
    hadamard_at q2
    hadamard_at q3
    qnot_at q3 `controlled` [q1, q2]
    hadamard_at q1
    hadamard_at q2
    gate_X_at q1
    gate_X_at q2
    hadamard_at q2
    qnot_at q2 `controlled` q1
    hadamard_at q2
    gate_X_at q1
    gate_X_at q2
    hadamard_at q1
    hadamard_at q2
    hadamard_at q3
    measure (q1,q2)
\end{lstlisting}
The first $3$ Hadamard gates are needed to obtain the linear superposition of the input qubits.
The CCNOT gate is the oracle.
The remaining gates, but the last two, implement the amplitude amplification steps.
The last Hadamard gate on the ancilla qubit performs the interference. Finally, the first 2 qubits are measured. 



\subsubsection{Translation and Validation}
Exploiting our implementation we automatically generate the code for QPMC shown in the Appendix.

According to the calculations we should reach the terminal state $s_{14}$ with probability equal to $1$, while the other terminal states must have an associated probability equal to $0$. 
We tested the formulas to evaluate the density matrix associated to each terminal state with input state $|1\rangle \langle 1|$ and the results are the following.

\begin{tabular}{p{5cm} p{6cm}}
  &  \\
\begin{lstlisting}[language=Haskell, breaklines=true, firstline=71, basicstyle=\tiny\ttfamily\linespread{0.5}, lastline=80]
qeval(Q=? [F (s = 11)], |1>_8 <1|_8);	

     0   0   0   0   0   0   0   0
     0  -0   0   0   0   0   0   0
     0   0   0   0   0   0   0   0
     0   0   0   0   0   0   0   0
     0   0   0   0   0   0   0   0
     0   0   0   0   0   0   0   0
     0   0   0   0   0   0   0   0
     0   0   0   0   0   0   0   0
\end{lstlisting} & 
\begin{lstlisting}[language=Haskell, breaklines=true, firstline=170, basicstyle=\tiny\ttfamily\linespread{0.5}, lastline=179]
qeval(Q=? [F (s = 12)], |1>_8 <1|_8);	

    0  0  0  0  0  0  0  0
    0  0  0  0  0  0  0  0
    0  0  0  0  0  0  0  0
    0  0  0  0  0  0  0  0
    0  0  0  0  0  0  0  0
    0  0  0  0  0  0  0  0
    0  0  0  0  0  0  0  0
    0  0  0  0  0  0  0  0
\end{lstlisting}

\end{tabular} 
\\
\begin{tabular}{p{5cm} p{6cm}}
  &  \\
\begin{lstlisting}[language=Haskell, breaklines=true, firstline=82, basicstyle=\tiny\ttfamily\linespread{0.5}, lastline=91]
qeval(Q=? [F (s = 13)], |1>_8 <1|_8);	

    0  0  0  0  0  0  0  0
    0  0  0  0  0  0  0  0
    0  0  0  0  0  0  0  0
    0  0  0  0  0  0  0  0
    0  0  0  0  0  0  0  0
    0  0  0  0  0  0  0  0
    0  0  0  0  0  0  0  0
    0  0  0  0  0  0  0  0
\end{lstlisting} & 
\begin{lstlisting}[language=Haskell, breaklines=true, firstline=183, basicstyle=\tiny\ttfamily\linespread{0.5}, lastline=192]
qeval(Q=? [F (s = 14)], |1>_8 <1|_8);	

     0   0   0   0   0   0   0   0
     0   0   0   0   0   0   0   0
     0   0   0   0   0   0   0   0
     0   0   0   0   0   0   0   0
     0   0   0   0   0   0   0   0
     0   0   0   0   0   0   0   0
     0   0   0   0   0   0   0  -0
     0   0   0   0   0   0   0   1
\end{lstlisting}

\end{tabular} 

%
%
%

It is possible to see that the trace of the first three matrices is equal to $0$, meaning that the probability of reaching those states is null. The density matrix associated to the last state has trace equal to $1$, meaning that the computation will surely reach that state, validating in this way the expected results.
We also tested formulas to calculate the accumulated superoperators for each state, but since the resulting matrices have size $2^6 \times 2^6$ we do not report them here. The results can be found at \url{https://github.com/miniBill/entangle}.

\subsection{Scalability of the swap algorithm}
We also decided to perfom some scalability tests on an artificial example which requires a high number of swaps. 
Recall that, since we need the circuit to be translated in SNF, for each Unitary gate we need to perform a number of binary swaps depending on the number $k$ of qubits used in the circuit.
In this part of the experiment we focused on the execution time of our implementation, i.e., the time required to produce the QPMC model.
The circuits given in input have been choosen to maximize the number of binary swaps required by our implementation. 
An example of such circuits of size $7$ can be seen in Figure \ref{fig:Test}.

\begin{figure}[H]
\begin{center}\includegraphics[scale=0.20]{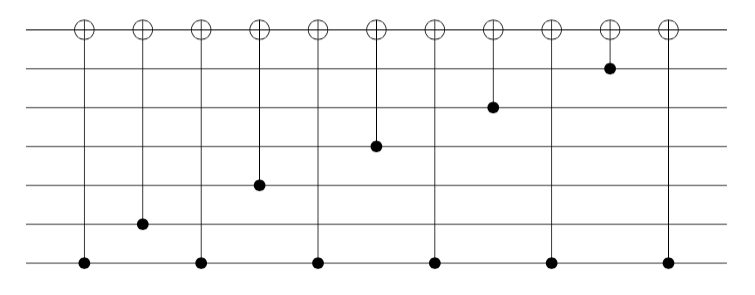}\end{center}
\protect\caption{Test circuit of size $7$}
\label{fig:Test}
\end{figure}

We decided to test circuits built using from $3$ to $8$ qubits.
Times are recorded using the \texttt{time} utility of the Bash shell on an early 2014 MacBook Air with a 1.4 GHz Intel Core i5 processor. For each size of the input, the program has been executed five times and the mean time has been computed. The results are shown in Figure \ref{fig:time}. 
We can see that also for a circuit of size $8$, when we have to generate swap matrices of size $2^8 \times 2^8$ our algorithm works in \emph{reasonable} times.
However, we are working on an improvement of our implementation in which we directly generate the swap matrices without having to compose binary swaps.
\begin{figure}[H]
\begin{center}\includegraphics[scale=0.45]{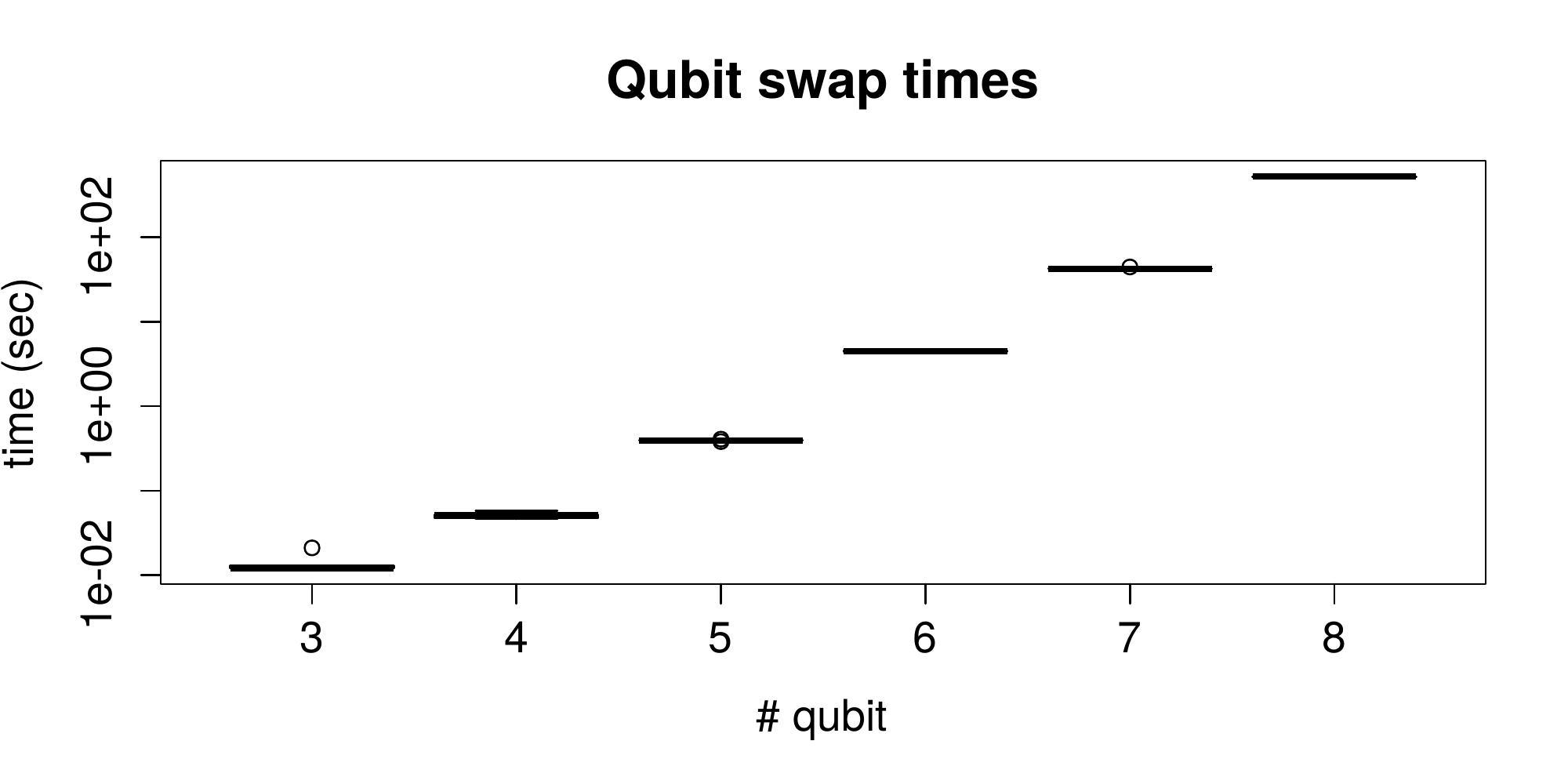}\end{center}
\protect\caption{Scalability test}
\label{fig:time}
\end{figure}

%% file: conclusion.tex
\section{Conclusion}\label{sec:conclusion}

In this work we proposed a framework that performs a translation from Quipper to QPMC. The main idea is to use this framework to create a tool that allows, on the one hand, the description of quantum algorithms and protocols in an high-level programming language, and on the other hand their formal verification.
In doing so we put particular attention in the translation at a semantic level. Quipper uses the state vector formalism and the quantum circuit model of computation while QPMC uses the density matrix formalism and QMC, allowing to consider also the measurements in the verification of the algorithms.
We implemented and tested our translator on some common quantum algorithms and the final results validated our expectations. 
We are working on enrichment and optimization of our framework in order to match the requirement of  validating complex algorithms and protocols, e.g., 
the ones involving also a classical control outside the \texttt{Circ} monad.
Moreover, we intend to  investigate  the specification of properties involving typical quantum and reversibility effects.



%% file: appendix.tex

\section*{Appendix}\label{sec:appendix}

Exploiting our implementation we automatically generate the following code for QPMC.
\begin{lstlisting}[language=Haskell, breaklines=true, firstline=22, basicstyle=\scriptsize\ttfamily\linespread{0.5}, lastline=68]
qmc
const matrix A1 = [(sqrt(2) / 4),(sqrt(2) / 4),(sqrt(2) / 4),(sqrt(2) / 4),(sqrt(2) / 4),(sqrt(2) / 4),(sqrt(2) / 4),(sqrt(2) / 4);(sqrt(2) / 4),((-1 * sqrt(2)) / 4),(sqrt(2) / 4),((-1 * sqrt(2)) / 4),(sqrt(2) / 4),((-1 * sqrt(2)) / 4),(sqrt(2) / 4),((-1 * sqrt(2)) / 4);(sqrt(2) / 4),(sqrt(2) / 4),((-1 * sqrt(2)) / 4),((-1 * sqrt(2)) / 4),(sqrt(2) / 4),(sqrt(2) / 4),((-1 * sqrt(2)) / 4),((-1 * sqrt(2)) / 4);(sqrt(2) / 4),((-1 * sqrt(2)) / 4),((-1 * sqrt(2)) / 4),(sqrt(2) / 4),(sqrt(2) / 4),((-1 * sqrt(2)) / 4),((-1 * sqrt(2)) / 4),(sqrt(2) / 4);(sqrt(2) / 4),(sqrt(2) / 4),(sqrt(2) / 4),(sqrt(2) / 4),((-1 * sqrt(2)) / 4),((-1 * sqrt(2)) / 4),((-1 * sqrt(2)) / 4),((-1 * sqrt(2)) / 4);(sqrt(2) / 4),((-1 * sqrt(2)) / 4),(sqrt(2) / 4),((-1 * sqrt(2)) / 4),((-1 * sqrt(2)) / 4),(sqrt(2) / 4),((-1 * sqrt(2)) / 4),(sqrt(2) / 4);(sqrt(2) / 4),(sqrt(2) / 4),((-1 * sqrt(2)) / 4),((-1 * sqrt(2)) / 4),((-1 * sqrt(2)) / 4),((-1 * sqrt(2)) / 4),(sqrt(2) / 4),(sqrt(2) / 4);(sqrt(2) / 4),((-1 * sqrt(2)) / 4),((-1 * sqrt(2)) / 4),(sqrt(2) / 4),((-1 * sqrt(2)) / 4),(sqrt(2) / 4),(sqrt(2) / 4),((-1 * sqrt(2)) / 4)];
const matrix A2 = [1,0,0,0,0,0,0,0;0,1,0,0,0,0,0,0;0,0,1,0,0,0,0,0;
	0,0,0,1,0,0,0,0;0,0,0,0,1,0,0,0;0,0,0,0,0,1,0,0;
	0,0,0,0,0,0,0,1;0,0,0,0,0,0,1,0];
const matrix A3 = [(sqrt(2) / 4),(sqrt(2) / 4),(sqrt(2) / 4),(sqrt(2) / 4),(sqrt(2) / 4),(sqrt(2) / 4),(sqrt(2) / 4),(sqrt(2) / 4);(sqrt(2) / 4),((-1 * sqrt(2)) / 4),(sqrt(2) / 4),((-1 * sqrt(2)) / 4),(sqrt(2) / 4),((-1 * sqrt(2)) / 4),(sqrt(2) / 4),((-1 * sqrt(2)) / 4);(sqrt(2) / 4),(sqrt(2) / 4),((-1 * sqrt(2)) / 4),((-1 * sqrt(2)) / 4),(sqrt(2) / 4),(sqrt(2) / 4),((-1 * sqrt(2)) / 4),((-1 * sqrt(2)) / 4);(sqrt(2) / 4),((-1 * sqrt(2)) / 4),((-1 * sqrt(2)) / 4),(sqrt(2) / 4),(sqrt(2) / 4),((-1 * sqrt(2)) / 4),((-1 * sqrt(2)) / 4),(sqrt(2) / 4);(sqrt(2) / 4),(sqrt(2) / 4),(sqrt(2) / 4),(sqrt(2) / 4),((-1 * sqrt(2)) / 4),((-1 * sqrt(2)) / 4),((-1 * sqrt(2)) / 4),((-1 * sqrt(2)) / 4);(sqrt(2) / 4),((-1 * sqrt(2)) / 4),(sqrt(2) / 4),((-1 * sqrt(2)) / 4),((-1 * sqrt(2)) / 4),(sqrt(2) / 4),((-1 * sqrt(2)) / 4),(sqrt(2) / 4);(sqrt(2) / 4),(sqrt(2) / 4),((-1 * sqrt(2)) / 4),((-1 * sqrt(2)) / 4),((-1 * sqrt(2)) / 4),((-1 * sqrt(2)) / 4),(sqrt(2) / 4),(sqrt(2) / 4);(sqrt(2) / 4),((-1 * sqrt(2)) / 4),((-1 * sqrt(2)) / 4),(sqrt(2) / 4),((-1 * sqrt(2)) / 4),(sqrt(2) / 4),(sqrt(2) / 4),((-1 * sqrt(2)) / 4)];
const matrix A4 = [0,0,0,0,0,0,1,0;0,0,0,0,0,0,0,1;0,0,0,0,1,0,0,0;
	0,0,0,0,0,1,0,0;0,0,1,0,0,0,0,0;0,0,0,1,0,0,0,0;
	1,0,0,0,0,0,0,0;0,1,0,0,0,0,0,0];
const matrix A5 = [(sqrt(2) / 2),0,(sqrt(2) / 2),0,0,0,0,0;0,(sqrt(2) / 2),0,(sqrt(2) / 2),0,0,0,0;(sqrt(2) / 2),0,((-1 * sqrt(2)) / 2),0,0,0,0,0;0,(sqrt(2) / 2),0,((-1 * sqrt(2)) / 2),0,0,0,0;0,0,0,0,(sqrt(2) / 2),0,(sqrt(2) / 2),0;0,0,0,0,0,(sqrt(2) / 2),0,(sqrt(2) / 2);0,0,0,0,(sqrt(2) / 2),0,((-1 * sqrt(2)) / 2),0;0,0,0,0,0,(sqrt(2) / 2),0,((-1 * sqrt(2)) / 2)];
const matrix A6 = [1,0,0,0,0,0,0,0;0,1,0,0,0,0,0,0;0,0,1,0,0,0,0,0;
	0,0,0,1,0,0,0,0;0,0,0,0,0,0,1,0;0,0,0,0,0,0,0,1;
	0,0,0,0,1,0,0,0;0,0,0,0,0,1,0,0];
const matrix A7 = [0,0,0,0,(sqrt(2) / 2),0,(sqrt(2) / 2),0;0,0,0,0,0,(sqrt(2) / 2),0,(sqrt(2) / 2);0,0,0,0,(sqrt(2) / 2),0,((-1 * sqrt(2)) / 2),0;0,0,0,0,0,(sqrt(2) / 2),0,((-1 * sqrt(2)) / 2);(sqrt(2) / 2),0,(sqrt(2) / 2),0,0,0,0,0;0,(sqrt(2) / 2),0,(sqrt(2) / 2),0,0,0,0;(sqrt(2) / 2),0,((-1 * sqrt(2)) / 2),0,0,0,0,0;0,(sqrt(2) / 2),0,((-1 * sqrt(2)) / 2),0,0,0,0];
const matrix A8 = [0,0,(sqrt(2) / 2),0,0,0,(sqrt(2) / 2),0;0,0,0,(sqrt(2) / 2),0,0,0,(sqrt(2) / 2);(sqrt(2) / 2),0,0,0,(sqrt(2) / 2),0,0,0;0,(sqrt(2) / 2),0,0,0,(sqrt(2) / 2),0,0;0,0,(sqrt(2) / 2),0,0,0,((-1 * sqrt(2)) / 2),0;0,0,0,(sqrt(2) / 2),0,0,0,((-1 * sqrt(2)) / 2);(sqrt(2) / 2),0,0,0,((-1 * sqrt(2)) / 2),0,0,0;0,(sqrt(2) / 2),0,0,0,((-1 * sqrt(2)) / 2),0,0];
const matrix A9 = [(sqrt(2) / 2),0,(sqrt(2) / 2),0,0,0,0,0;0,(sqrt(2) / 2),0,(sqrt(2) / 2),0,0,0,0;(sqrt(2) / 2),0,((-1 * sqrt(2)) / 2),0,0,0,0,0;0,(sqrt(2) / 2),0,((-1 * sqrt(2)) / 2),0,0,0,0;0,0,0,0,0,0,0,0;0,0,0,0,0,0,0,0;0,0,0,0,0,0,0,0;0,0,0,0,0,0,0,0];
const matrix A10 = [0,0,0,0,0,0,0,0;0,0,0,0,0,0,0,0;0,0,0,0,0,0,0,0;0,0,0,0,0,0,0,0;0,0,0,0,(sqrt(2) / 2),0,(sqrt(2) / 2),0;0,0,0,0,0,(sqrt(2) / 2),0,(sqrt(2) / 2);0,0,0,0,(sqrt(2) / 2),0,((-1 * sqrt(2)) / 2),0;0,0,0,0,0,(sqrt(2) / 2),0,((-1 * sqrt(2)) / 2)];
const matrix A11 = [1,0,0,0,0,0,0,0;0,1,0,0,0,0,0,0;0,0,0,0,0,0,0,0;
	0,0,0,0,0,0,0,0;0,0,0,0,1,0,0,0;0,0,0,0,0,1,0,0;
	0,0,0,0,0,0,0,0;0,0,0,0,0,0,0,0];
const matrix A12 = [0,0,0,0,0,0,0,0;0,0,0,0,0,0,0,0;0,0,1,0,0,0,0,0;
0,0,0,1,0,0,0,0;0,0,0,0,0,0,0,0;0,0,0,0,0,0,0,0;
0,0,0,0,0,0,1,0;0,0,0,0,0,0,0,1];
const matrix A13 = [1,0,0,0,0,0,0,0;0,1,0,0,0,0,0,0;0,0,0,0,0,0,0,0;
	0,0,0,0,0,0,0,0;0,0,0,0,1,0,0,0;0,0,0,0,0,1,0,0;
    	0,0,0,0,0,0,0,0;0,0,0,0,0,0,0,0];
const matrix A14 = [0,0,0,0,0,0,0,0;0,0,0,0,0,0,0,0;0,0,1,0,0,0,0,0;
	0,0,0,1,0,0,0,0;0,0,0,0,0,0,0,0;0,0,0,0,0,0,0,0;
    	0,0,0,0,0,0,1,0;0,0,0,0,0,0,0,1];
module test
  s: [0..14] init 0;
  [] (s = 0) -> <<A1>> : (s' = 1);
  [] (s = 1) -> <<A2>> : (s' = 2);
  [] (s = 2) -> <<A3>> : (s' = 3);
  [] (s = 3) -> <<A4>> : (s' = 4);
  [] (s = 4) -> <<A5>> : (s' = 5);
  [] (s = 5) -> <<A6>> : (s' = 6);
  [] (s = 6) -> <<A7>> : (s' = 7);
  [] (s = 7) -> <<A8>> : (s' = 8);
  [] (s = 8) -> <<A9>> : (s' = 9) + <<A10>> : (s' = 10);
  [] (s = 9) -> <<A11>> : (s' = 11) + <<A12>> : (s' = 12);
  [] (s = 10) -> <<A13>> : (s' = 13) + <<A14>> : (s' = 14);
  [] (s = 11) -> (s' = 11);
  [] (s = 12) -> (s' = 12);
  [] (s = 13) -> (s' = 13);
  [] (s = 14) -> (s' = 14);
endmodule
\end{lstlisting}

%% file: main.bbl
\begin{thebibliography}{1}

\bibitem{QPMC}
Y.~Feng, E.~M. Hahn, A.~Turrini, and L.~Zhang.
\newblock Qpmc: A model checker for quantum programs and protocols.
\newblock In Nikolaj Bj{\o}rner and Frank~D. de~Boer, editors, {\em FM 2015:
  Formal Methods - 20th International Symposium, Oslo, June 24-26, 2015,
  Proceedings}, Lecture Notes in Computer Science. Springer, 2015.

\bibitem{MCQMC2013}
Y.~Feng, N.~Yu, and M.~Ying.
\newblock Model checking quantum markov chains.
\newblock {\em Journal of Computer and System Sciences}, 2013.

\bibitem{Selinger2013}
A.S. Green, P.L. Lumsdaine, N.J. Ross, P.~Selinger, and B.~Valiron.
\newblock Quipper: A scalable quantum programming language.
\newblock {\em SIGPLAN Not.}, 48(6), 2013.

\bibitem{knill-qram}
E.~Knill.
\newblock Conventions for quantum pseudocode.
\newblock Technical report, Los Alamos National Laboratory, 1996.

\bibitem{PRISM}
M.~Kwiatkowska, G.~Norman, and D.~Parker.
\newblock Prism 4.0: Verification of probabilistic real- time systems.
\newblock In G.~Gopalakrishnan and S.~Qadeer, editors, {\em LNCS}, volume 6806,
  2011.

\bibitem{Nielsen-Chuang}
M.A. Nielsen and I.L. Chuang.
\newblock {\em Quantum Computation and Quantum Information}.
\newblock Cambridge University Press, 2011.

\bibitem{Preskill}
J.~Preskill.
\newblock {\em Lecture Notes for Physics 229: Quantum Information and
  Computation}.
\newblock CreateSpace Independent Publishing Platform, 1998.

\bibitem{Selinger2014}
J.M. Smith, N.J. Ross, P.~Selinger, and B.~Valiron.
\newblock Quipper: concrete resource estimation in quantum algorithms.
\newblock Extended abstract for a talk given at the 12th International Workshop
  on Quantitative Aspects of Programming Languages and Systems, QAPL 2014,
  Grenoble. Available from arxiv1412.0625, 2014.

\end{thebibliography}
